\documentclass[conference,a4paper]{IEEEtran}

\usepackage{amsmath,amsthm,amsfonts,braket,algorithm,graphicx}

\theoremstyle{definition} \newtheorem{define} {Definition} [section]
\theoremstyle{definition} 
\newtheorem {theorem} {Theorem}

\newtheorem {lemma} {Lemma}

\newcommand{\kb}[1]{\mathbf{\left[#1\right]}}

\newcommand{\trd}[1]{\left|\left| #1 \right| \right|}

\newcommand{\MR}{\texttt{Measure and Resend}}
\newcommand{\R}{\texttt{Reflect}}

\newcommand{\rew}{\texttt{\textbf{Rw}}}

\newcommand{\protpm}{\Pi^{\texttt{SQKD}}}
\newcommand{\protp}{\widetilde{\Pi}^{\texttt{SQKD}}}
\newcommand{\protent}{\Pi^{*}}

\newcommand{\conj}[1]{#1^*}

\floatname{algorithm}{Protocol}

\begin{document}
\title{Key-Rate Bound of a Semi-Quantum Protocol Using an Entropic Uncertainty Relation}
\author{
\IEEEauthorblockN{Walter O. Krawec}
\IEEEauthorblockA{Computer Science \& Engineering Department\\
University of Connecticut\\
Storrs, CT 06268\\
Email: walter.krawec@uconn.edu}
}

\maketitle

\begin{abstract}
In this paper we present a new proof technique for semi-quantum key distribution protocols which makes use of a quantum entropic uncertainty relation to bound an adversary's information.  Our new technique provides a more optimistic key-rate bound than previous work relying only on noise statistics (as opposed to using additional mismatched measurements which increase the noise tolerance of this protocol, but at the cost of requiring four times the amount of measurement data).  Our new technique may hold application in the proof of security of other semi-quantum protocols or protocols relying on two-way quantum communication.
\end{abstract}

\section{Introduction}

Quantum Key Distribution (QKD) protocols allow for the establishment of a secret key between two parties Alice ($A$) and Bob ($B$) which is secure even against an all-powerful adversary Eve ($E$) - that is, an adversary bounded only by the laws of physics and not by any computational assumptions as is required when using classical communication only.  Numerous QKD protocols have been proposed, many with rigorous proofs of unconditional security.  The reader is referred to \cite{QKD-survey} for a general survey.

However, these QKD protocols, and their security analysis, require both $A$ and $B$ to be ``quantum capable.''  Namely, both $A$ and $B$ must be equipped with devices capable of manipulating quantum resources in certain, arbitrary, ways (e.g., preparing and measuring qubits in certain bases).  In 2007, Boyer et al., in \cite{SQKD-first-PRL} introduced the \emph{semi-quantum} model of cryptography whereby only $A$ was required to be quantum while $B$ was allowed to be very limited and ``classical'' in his capabilities.  These semi-quantum key distribution (SQKD) protocols are interesting to study theoretically as they attempt to answer the question `` how quantum does a protocol need to be to gain an advantage over its classical counterpart?''  There are also potential practical benefits to studying these protocols: for example, $B$'s device could be cheaper to manufacture; alternatively, one can consider designing a QKD infrastructure more robust to technical failure - indeed, if a device ever breaks down, one may switch to a ``semi-quantum'' mode and continue secure operations until the device is fully repaired.

SQKD protocols, however, require a two-way quantum communication channel (one which allows $A$ to send qubits to $B$ who then sends qubits back to $A$) greatly increasing the complexity of their security analysis.  Though several SQKD protocols have been proposed (see \cite{SQKD-first-PRL,SQKD-second,SQKD-lessthan4,SQKD-Single-Security,SQKD-MultiUser,SQKD-3,SQKD-cl-A,SQKD-4,SQKD-no-measure1,SQKD-no-measure2,SQKD-mirror} for a few), until 2015, most were proven only to be \emph{robust} - a notion introduced in \cite{SQKD-first-PRL} which says any attack which causes $E$ to gain non-zero information on the key must induce a disturbance which can be detected with non-zero probability.  Some authors considered security against individual attacks \cite{SQKD-information} (attacks whereby $E$ is forced to measure her quantum ancilla before the protocol concludes).  It wasn't until 2015, that rigorous proofs of security became available in \cite{SQKD-Krawec-SecurityProof,SQKD-Krawec-ReflectionSecurity,SQKD-zhang2016single}.

In a recent work \cite{QKD-Tom-Krawec-Arbitrary}, we showed that the original SQKD protocol of Boyer et al., has a noise tolerance of $11\%$ - exactly the same as the fully-quantum BB84 protocol.  Our result in \cite{QKD-Tom-Krawec-Arbitrary}, however, required the use of numerous measurements, including \emph{mismatched measurements} \cite{QKD-Tom-First,QKD-Tom-KeyRateIncrease}.  Ultimately, to compute the key-rate of the Boyer et al., protocol, using our technique in that paper, one must look at over 12 different measurement statistics and then evaluate a series of lengthy equations. (Indeed, our key-rate equation for the SQKD protocol spanned numerous pages!)

In this work, we revisit this semi-quantum protocol and derive a simpler, and far more elegant (in the author's opinion) proof of security using a quantum uncertainty bound to evaluate the von Neumann entropy of the resulting quantum system.  Our new bound does not use mismatched measurements (only error rates) and, so, the noise tolerance is not as high as in \cite{QKD-Tom-Krawec-Arbitrary}; however our new result is higher than previous best-known results for this protocol without mismatched measurements.  Furthermore, the technique we present here may be simpler to adapt to other SQKD protocols than the technique using mismatched measurements - especially for higher-dimensional protocols (such as \cite{QKD-vlachou2017quantum}) where the technique of mismatched measurements can become intractable.

There are several contributions made in this work, many of which we expect would hold great application outside the scope of this paper.  First, we show that for \emph{any} semi-quantum protocol, it is sufficient to consider a restricted form of collective attack.  Second, we show an entirely new approach to proving security of semi-quantum protocols; we show how to convert a particular SQKD protocol into an equivalent entanglement based version and we derive a new key-rate bound which does not require the use of numerous mismatched measurement statistics and which produces a higher noise tolerance than previous work without these statistics (along with a far simpler key-rate expression).  Note that, in \cite{QKD-twoway2}, a technique of converting certain two-way QKD protocols into equivalent entanglement based versions was shown; however their result could only be applied to protocols where $B$'s output is independent of his input averaged over all of his operations - this property is sadly lacking in the semi-quantum model and so a new method is required which we introduce in this paper.  Third, our proof shows a new and interesting application of a quantum uncertainty bound to the semi-quantum model of cryptography and also an interesting application of a continuity bound on conditional von Neumann entropy which may be of great use when proving security of new protocols in the semi-quantum model (or, more generally, for protocols relying on a two-way channel which may not hold certain symmetry properties).

\subsection{Notation}
We assume the reader is familiar with basic quantum information theory and so here we will only introduce our notation and a few general concepts; for a general survey see \cite{QC-Intro}.  The \emph{computational $Z$ basis} is defined to be $\{\ket{0}, \ket{1}\}$ while the \emph{Hadamard $X$ basis} is $\{\ket{+}, \ket{-}\}$, where $\ket{\pm} = \frac{1}{\sqrt{2}}(\ket{0}\pm\ket{1})$.

We denote by $H(p_1, \cdots, p_n)$ to be the Shannon entropy of $p_1, \cdots, p_n$.  If $A$ and $B$ are random variables, then $H(A|B)$ is the conditional Shannon entropy of $A$ conditioned on $B$.  By $h(x)$ we mean the binary entropy function: $h(x) = H(x, 1-x) = -x\log x - (1-x)\log (1-x)$.  All logarithms in this paper are base two.

A \emph{density operator} is a Hermitian positive semi-definite operator of unit trace.  If $\rho$ is a density operator acting on Hilbert space $\mathcal{H}_A\otimes\mathcal{H}_B$, we often write $\rho_{AB}$.  In this case, we write $\rho_B$ to mean the operator resulting from tracing out $A$; i.e., $\rho_B = tr_A\rho_{AB}$.  Similarly when the operator acts on larger systems.  We also will write $\mathcal{H}_{AB}$ to denote $\mathcal{H}_{A}\otimes\mathcal{H}_{B}$.

Given density operator $\rho_{AB}$ we write $S(AB)_\rho$ to mean the von Neumann entropy of $\rho_{AB}$.  We write $S(A|B)_\rho$ to mean the conditional von Neumann entropy: $S(A|B)_\rho = S(AB)_\rho - S(B)_\rho$.  If the context is clear, we will forgo writing the subscript ``$\rho$.''

Given an operator $A$, we write $\trd{A}$ to mean the trace norm of $A$.  If $A$ is Hermitian and finite dimensional, then $\trd{A} = \sum_i|\lambda_i|$, where $\{\lambda_i\}$ are the eigenvalues of $A$.

If $z \in \mathbb{C}^{m\times n}$, then we write $z^*$ to mean the conjugate transpose of $z$.  Also, we define $\mathbb{D} = \{z\in\mathbb{C} \text{ } | \text{ } |z| \le 1\}$.  Finally, we write $\kb{a,b}_{AB}$ to mean $\ket{a,b}\bra{a,b}_{AB}$.

\subsection{(S)QKD Security}
A (S)QKD protocol operates in two stages: first a \emph{quantum communication stage} whereby users utilize the communication resources available to them (typically a quantum channel and an authenticated classical channel) to establish a \emph{raw-key} which is a string of $0$'s and $1$'s which is partially correlated and partially secret.  In general, this yields a classical-classical-quantum state of the form:
\begin{equation}\label{eq:ccq-state}
\rho_{ABE} = \sum_{a,b\in\{0,1\}^n}\kb{a,b}_{AB}\otimes\rho_E^{(a,b)},
\end{equation}
where the $A$ and $B$ register represent $A$ and $B$'s raw-key respectively, while $\rho_E^{(a,b)}$ (which is not necessarily of unit trace) represents the state of $E$'s memory in the event $A$ and $B$ have a raw-key of $a$ and $b$ respectively.  Following the quantum communication stage, a classical stage consisting of error correction and privacy amplification is run producing a secret key of size $\ell(n)$ bits which may then be used for other cryptographic protocols.

An important question is, given certain observations on $E$'s attack (e.g., the error rate), how large is $\ell(n)$?  For collective attacks (attacks whereby $E$ performs the same attack operation each iteration however is free to postpone measurement of her ancilla until any future time of her choice), Equation \ref{eq:ccq-state} takes on the simpler form:
\[
\rho_{ABE} = \left(\sum_{a,b\in\{0,1\}}\kb{a,b}_{AB}\otimes\rho_E^{(a,b)}\right)^{\otimes n},
\]
in which case the Devetak-Winter key-rate expression \cite{QKD-winter-keyrate} may be employed which states:
\begin{equation}\label{eq:keyrate}
r := \lim_{n\rightarrow \infty} \frac{\ell(n)}{n} = \inf[S(A|E) - H(A|B)],
\end{equation}
where the infimum is over all collective attacks which induce the observed noise statistics.  Computing a bound on this expression is the key critical element of any (S)QKD security proof \cite{QKD-survey}.

\section{The Protocol}

The protocol we consider is the original SQKD protocol of Boyer et al., introduced in \cite{SQKD-first-PRL,SQKD-second}.  This protocol, being a semi-quantum one, assumes that $A$ is \emph{fully quantum} in that she can prepare and measure qubits in arbitrary bases; however $B$ is \emph{classical} in that he can only directly work with the computational $Z$ basis.  In more detail, a SQKD protocol utilizes a two-way quantum channel.  We call the channel connecting $A$ to $B$ the \emph{forward channel} and the channel connecting $B$ to $A$ the \emph{reverse channel}.  Each iteration of the quantum communication stage, $A$ will prepare and send a qubit.  $B$ is then restricted to choosing between two operations: $\MR$ or to $\R$.  If he chooses $\MR$, he will subject the incoming qubit to a $Z$ basis measurement and prepare a new qubit in the same state he observed (i.e., if his measurement result is $\ket{r}$, for $r \in \{0,1\}$, he will send a qubit $\ket{r}$ back to $A$); if he chooses $\R$ he will completely disconnect from the quantum channel, allowing the qubit to pass through his lab undisturbed, and return to $A$ (in essence, if $B$ chooses $\R$, $A$ is ``talking to herself'').

The protocol we consider, and which we denote as $\protpm$, operates as follows:
\begin{enumerate}
  \item $A$ chooses to send a qubit of the form $\ket{0}$, $\ket{1}$, $\ket{+}$, or $\ket{-}$, choosing randomly.
  \item $B$ will choose to $\MR$ or to $\R$.  If he chooses $\MR$ he will save his measurement result in a classical register to serve as his potential raw-key bit for this iteration.
  \item $A$ will measure in the same basis she originally used to send.
  \item $A$ will disclose her choice of basis; $B$ will disclose his choice of operation ($\MR$ or $\R$).  This is done using an authenticated classical channel.
  \item If $A$ used the $Z$ basis and if $B$ chose to $\MR$, then they will use this iteration to contribute towards their raw-key.  $A$ will use her initial preparation choice as her key bit (equivalently, she may use her measurement result at the end - our security analysis will apply to both cases).  Other iterations, along with a suitably sized random subset of these ``raw-key'' iterations, may be used to estimate the error rate in the channel.  In particular, $A$ and $B$ may estimate the $Z$ basis error rate in the forward, reverse, and joint channel.  They may also observe the $X$ basis error rate in the joint channel (but not the forward or reverse separately since $B$ is unable to prepare or measure in the $X$ basis).
\end{enumerate}

It is not difficult to see that the protocol is correct.  We analyze its security by determining a new lower-bound on the Devetak-Winter key-rate expression (Equation \ref{eq:keyrate}) for this protocol.

\section{Security Proof}
We prove security against collective attacks in this paper - we will comment on general attacks later.  Collective attacks are those where $E$ applies the same attack operation each iteration of the protocol, but is free to postpone measurement of her ancilla until any future time of her choice.  In the semi-quantum model, where an attacker has two opportunities to attack a qubit each iteration, a collective attack is a pair of unitary operators $(U_F, U_R)$, each acting on the Hilbert space $\mathcal{H}_T\otimes\mathcal{H}_E$ (here $\mathcal{H}_T$ is the two-dimensional space modeling the qubit ``transit'' space while $\mathcal{H}_E$ is $E$'s quantum ancilla).  The operator $U_F$ is applied in the Forward channel (as a qubit travels from $A$ to $B$) while $U_R$ is applied in the Reverse channel (when the qubit returns from $B$ to $A$).

Our proof of security follows three steps.  First, we will prove that for \emph{any} semi-quantum protocol, it is sufficient to consider a particular ``restricted'' collective attack which is easier to analyze mathematically, but does not cause $E$ to lose attack power.  Second, using this result, we show how to convert the protocol of interest into a mathematically equivalent entanglement based version .  Third, we use a quantum uncertainty bound and a continuity bound on conditional entropy to analyze the entanglement based version - security of the SQKD protocol will then follow.  See Figure \ref{fig:reductions}.

\begin{figure}
  \centering
  \includegraphics[width=220pt]{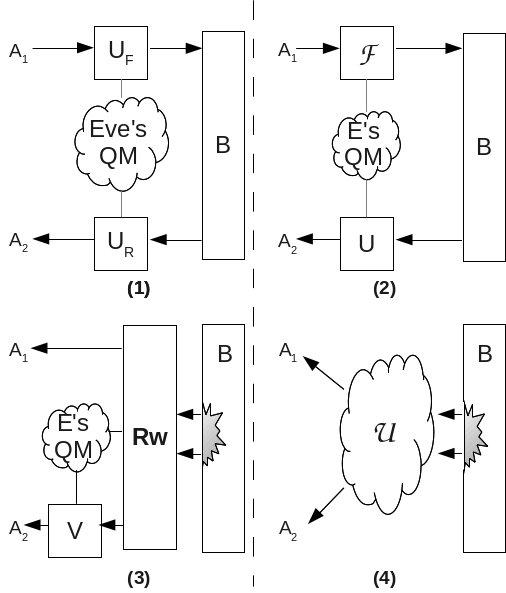}
\caption{Our security proof in pictures (note: QM = ``Quantum Memory'').  (1) is an arbitrary collective attack against the actual SQKD protocol (with classical Bob) that we wish to analyze.  We show that for every such attack, there is an equivalent \emph{restricted} attack (2).  Next, we show that, mathematically, analyzing (2) is the same as analyzing a different protocol where a quantum Bob prepares pairs of qubits (and his qubit source is partially influenced by Eve, represented here as a shaded region in $B$'s lab) and sends them to Alice; namely, for every restricted attack against the SQKD protocol, there is an equivalent attack against our new protocol where first Eve ``rewinds'' the channel state with operator $\rew$ (simulating the system had $A$ sent a qubit first instead of $B$), and then attacks one of the qubits.  Of course analyzing (4), where Eve is not restricted to performing this ``rewind'' attack, only gives her more power.  Thus security in (4) $\Rightarrow$ (3) $\Rightarrow$ (2) $\Rightarrow$ (1), our original goal.}\label{fig:reductions}
\end{figure}

\subsection{Restricted Attacks}

In \cite{SQKD-Single-Security}, we showed that for any \emph{single-state} SQKD protocol (i.e., those where $A$ is restricted to sending a single, publicly known, qubit state each iteration, typically $\ket{+}$ \cite{SQKD-lessthan4}), to prove security against collective attacks it is sufficient to consider security against a ``restricted'' attack whereby Eve, in the forward channel, need only bias Bob's measurement result; in the reverse channel, she applies an arbitrary unitary operator.  That is, it is not required that she perform an arbitrary unitary operator in the forward channel, entangling the qubit with her private quantum memory.  Such attacks are easier to analyze - and have been used in \cite{SQKD-Krawec-ReflectionSecurity,SQKD-zhang2016single} to show security of several different single-state protocols - however, as shown in \cite{SQKD-Single-Security}, the result is only correct for single state protocols.  The original SQKD protocol of Boyer et al., which we are considering in this work (i.e., $\protpm$), is a \emph{multi-state} protocol, one where $A$ prepares different states each iteration, choosing randomly each time.  However, it remained an open question as to whether or not some other form of restricted attack might be constructed for multi-state protocols.  We answer this question in the affirmative.

\begin{define}\label{def:rest-atk}
Let $\mathcal{B} = \{\ket{v_0}, \ket{v_1}\}$ be an orthonormal basis.  A \emph{multi-state restricted collective attack with respect to $\mathcal{B}$} is a tuple $(q_0, q_1, \eta_0, \eta_1, U)$, where $q_0, q_1 \in [0,1]$; $\eta_0, \eta_1 \in \mathbb{D}$ subject to the restriction that $q_0\eta_1\sqrt{1-q_1^2} + q_1\conj{\eta_0}\sqrt{1-q_0^2} = 0$; and $U$ is a unitary operator acting on $\mathcal{H}_T\otimes\mathcal{H}_E$.  The attack consists of the following actions:
\begin{enumerate}
  \item When $E$ first captures the qubit from $A$ in the forward channel, she applies the operator $\mathcal{F}$, acting on $\mathcal{H}_T\otimes\mathcal{H}_E$ which acts as follows:
\begin{align}
\mathcal{F}\ket{v_0} &= q_0\ket{0,0}_{TE} + \sqrt{1-q_0^2}\ket{1,e}_{TE}\label{eq:rest-atk-forward}\\
\mathcal{F}\ket{v_1} &= \sqrt{1-q_1^2}\ket{0,f}_{TE} + q_1\ket{1,0}_{TE},\notag
\end{align}
where:
\begin{align}
\ket{e} &= \eta_0\ket{0}_E + \sqrt{1-|\eta_0|^2}\ket{1}_E\label{eq:rest-atk-e}\\
\ket{f} &= \eta_1\ket{0}_E + \sqrt{1-|\eta_1|^2}\ket{1}_E.\label{eq:rest-atk-f}
\end{align}
Note that it is not difficult to see, considering the restrictions on the values $\eta_0, \eta_1, q_0,$ and $q_1$, that $\mathcal{F}$ is an isometry and may therefore be easily extended to a unitary operator; thus this is an operation $E$ can perform within the laws of quantum physics.
\item When the qubit returns from $B$ (on its way back to $A$ in the reverse channel), $E$ captures it again, and applies the unitary operator $U$.
\end{enumerate}
\end{define}

When the context is clear, we will simply call the above attack a \emph{restricted attack} as opposed to its longer title.  The following theorem proves that it is sufficient to consider these restricted attacks when proving security of \emph{any} semi-quantum protocol against arbitrary collective attacks.

\begin{theorem}\label{thm:rest-atk}
Let $\mathcal{B}=\{\ket{v_0},\ket{v_1}\}$ be an arbitrary orthonormal basis. For every collective attack $\mathcal{C} = (U_F, U_R)$, there exists a restricted attack $\mathcal{R} = (q_0, q_1, \eta_0, \eta_1, U)$ such that, for any SQKD protocol with quantum $A$ and classical $B$, the following are true:
\begin{enumerate}
  \item $A$ and $B$ cannot distinguish between attack $\mathcal{C}$ and $\mathcal{R}$.
  \item $E$'s final quantum system is the same regardless of whether she used $\mathcal{C}$ or $\mathcal{R}$.
  \item The key-rate is equal under both attacks.
\end{enumerate}
\end{theorem}
\begin{proof}
Let $\mathcal{B}$ be given and fix a collective attack $\mathcal{C}=(U_F, U_R)$, where $U_F$ is a unitary operator applied in the forward channel while $U_R$ is a unitary operator applied in the reverse.  We will construct a restricted attack $\mathcal{R}$ satisfying the required conditions.  Without loss of generality, we may assume $E$'s system is cleared to some ``zero'' state $\ket{\chi}_E$ at the start of the attack.  In this case, we may write $U_F$'s action on basis states as follows:
\begin{align*}
U_F\ket{v_0,\chi} &= \alpha\ket{0,e_0}_{TE} + \sqrt{1-\alpha^2}\ket{1,e_1}_{TE}\\
U_F\ket{v_1,\chi} &= \sqrt{1-\beta^2}\ket{0,e_2}_{TE} + \beta\ket{1,e_3}_{TE}
\end{align*}
where $\alpha,\beta\in[0,1]$ (any phase may be absorbed into the vectors $\ket{e_i}$) and the $\ket{e_i}$ are arbitrary normalized (though not necessarily orthogonal) states in $\mathcal{H}_E$.  Let $q_0 = \alpha$ and $q_1 = \beta$ for our restricted attack.

Unitarity of $U_F$ imposes the following condition:
\[
\alpha\sqrt{1-\beta^2}\braket{e_0|e_2} + \beta\sqrt{1-\alpha^2}\braket{e_1|e_3} = 0.
\]
Let $\eta_0 = \braket{e_3|e_1}$ and $\eta_1 = \braket{e_0|e_2}$ for our restricted attack $\mathcal{R}$ (this clearly satisfies the required restrictions on $q_i$ and $\eta_i$).  Now, first consider the case that $|\eta_i| < 1$ (the case when $|\eta_i| = 1$ for one or both $i$ will be considered afterwards).  Let $V$ be the operator acting on states in $\mathcal{H}_T\otimes\mathcal{H}_E$ as follows:
\begin{align*}
V\ket{0,0}_{TE} = \ket{0,e_0} && V\ket{1,0}_{TE} = \ket{1,e_3}\\
V\ket{0,1}_{TE} = \ket{0,g_0} && V\ket{1,1}_{TE} = \ket{1,g_1}
\end{align*}
where:
\begin{align*}
\ket{g_0} = \frac{\ket{e_2} - \eta_1\ket{e_0}}{\sqrt{1-|\eta_1|^2}}, &&\ket{g_1} = \frac{\ket{e_1} - \eta_0\ket{e_3}}{\sqrt{1-|\eta_0|^2}}.
\end{align*}
We claim that $V$ is an isometry.  First, it is clear that each $\ket{g_i}$ is normalized.  Indeed:
\begin{align*}
\braket{g_0|g_0} &= \frac{1+|\eta_1|^2-2Re(\eta_1\braket{e_2|e_0})}{1-|\eta_1|^2}\\
&=\frac{1+|\eta_1|^2-2Re|\eta_1|^2}{1-|\eta_1|^2} = 1,
\end{align*}
where, above, we used the fact that $\eta_1 = \braket{e_0|e_2}$.  A similar computation yields $\braket{g_1|g_1} = 1$.  What remains to be shown is that $\braket{e_0|g_0} = \braket{e_3|g_1} = 0$.  But this is clear:
\[
\braket{e_0|g_0} = \frac{\braket{e_0|e_2} - \eta_1}{\sqrt{1-|\eta_0|^2}} = \frac{\eta_1 - \eta_1}{\sqrt{1-|\eta_0|^2}} = 0,
\]
and similarly for $\braket{e_3|g_1}$.  We conclude, therefore, that $V$ is an isometry which may be extended to a unitary operator (its action on states not shown above is irrelevant).  In the following, we will assume $V$ is simply unitary.

To simplify notation in the remainder of this proof, let $N_i = \sqrt{1-|\eta_i|^2}$.  By linearity of $V$, it follows that:
\begin{align*}
V\ket{1,e}_{TE} &= V(\eta_0\ket{1,0} + N_0\ket{1,1})\\
&= \eta_0\ket{1,e_3} + N_0\ket{1}\otimes\left(\frac{\ket{e_1} - \eta_0\ket{e_3}}{N_0}\right)\\
&= \ket{1,e_1},
\end{align*}
and similarly, $V\ket{0,f} = \ket{0,e_2}$ (where $\ket{e}$ and $\ket{f}$ are the states resulting from the action of the restricted attack operator $\mathcal{F}$ defined in Equations \ref{eq:rest-atk-e} and \ref{eq:rest-atk-f}).

Let $U = U_RV$.  We claim that $\mathcal{R} = (q_0,q_1,\eta_0,\eta_1,U)$, as constructed above, is the desired restricted attack.  (Note, we are still assuming for the time being, that $|\eta_i| < 1$.)

We first consider the case where $A$'s sent state is pure; for mixed states the result will follow immediately due to linearity of the operations.  On any particular iteration of the protocol, let $\ket{a} = x\ket{v_0} + y\ket{v_1}$ be the state prepared and sent by $A$.  These $x$ and $y$ are potentially known only to $A$ (i.e., $A$ may choose, each iteration, to send a randomly prepared state in which case the $x$ and $y$ are chosen randomly).  Let $\ket{\psi_R}$ be the state of the qubit as it arrives to $B$ if $E$ uses the restricted attack as constructed above.  Let $\ket{\psi_C}$ be the state if $E$ uses the collective attack $\mathcal{C} = (U_F, U_R)$.  These are both easily computed:
\begin{align}
\ket{\psi_R} = \mathcal{F}\ket{a} &= \ket{0}\otimes \left(xq_0\ket{0}_E + y\sqrt{1-q_1^2}\ket{f}_E\right)\label{eq:state-restricted}\\
&+\ket{1}\otimes\left(x\sqrt{1-q_0^2}\ket{e}_E + yq_1\ket{0}_E\right).\notag
\end{align}
\begin{align}
\ket{\psi_C} = U_F\ket{a,\chi} &= \ket{0}\otimes\left(x\alpha\ket{e_0} + y\sqrt{1-\beta^2}\ket{e_2}\right)\label{eq:state-collective}\\
&+\ket{1}\otimes\left(x\sqrt{1-\alpha^2}\ket{e_1} + y\beta\ket{e_3}\right).\notag
\end{align}

At this point, $B$ will either $\MR$ (saving his measurement result in a private register) or $\R$.  We first consider the case where he chooses $\MR$.  Let $\rho^R$ be the result of this operation in the restricted attack case and $\rho^C$ the result in the general collective attack case.  These density operators are found to be:
\begin{align*}
\rho^R &= \kb{0,0}_{BT}\otimes P\left(xq_0\ket{0}_E + y\sqrt{1-q_1^2}\ket{f}_E\right)\\
&+\kb{1,1}_{BT}\otimes P\left(x\sqrt{1-q_0^2}\ket{e}_E + yq_1\ket{0}_E\right)
\end{align*}
\begin{align*}
\rho^C &= \kb{0,0}_{BT}\otimes P\left(x\alpha\ket{e_0} + y\sqrt{1-\beta^2}\ket{e_2}\right)\\
&+\kb{1,1}_{BT}\otimes P\left(x\sqrt{1-\alpha^2}\ket{e_1} + y\beta\ket{e_3}\right),
\end{align*}
where $P(z) = zz^*$.  At this point, if $E$ is using the restricted attack, she will now apply the operator $U = U_RV$.  However, $V$'s action of $\rho^R$ evolves the state to:
\begin{align*}
V\rho^RV^* &= \kb{0,0}_{BT} \otimes P\left(x\alpha\ket{e_0} + y\sqrt{1-\beta^2}\ket{e_2}\right)\\
&+\kb{1,1}_{BT} \otimes P\left(x\sqrt{1-\alpha^2}\ket{e_1} + y\beta\ket{e_3}\right)\\
&=\rho^C.
\end{align*}
Thus, it is clear that $U\rho^RU^* = U_R\rho^CU_R^*$ and so the resulting quantum state is identical regardless of whether $E$ used the restricted attack $\mathcal{R}$ or the collective attack $\mathcal{C}$ ($A$, $B$, and $E$'s systems are identical regardless).

We now consider the case when $B$ reflects.  Let $\sigma^R$ be the state of the system when the qubit leaves $B$'s lab and $E$ used the restricted attack; let $\sigma^C$ be the state of the system if $E$ uses the collective attack.  That is: $\sigma^R = \ket{\psi_R}\bra{\psi_R}$ and $\sigma^C = \ket{\psi_C}\bra{\psi_C}$, where $\ket{\psi_R}$ and $\ket{\psi_C}$ are defined in Equations \ref{eq:state-restricted} and \ref{eq:state-collective}.  It is trivial to show that $V\sigma^RV^* = \sigma^C$ and so the result holds in the case $B$ reflects.

In the above, we assumed that $|\eta_i| < 1$ for both $i=0,1$.  However, if $|\eta_i| = 1$ for one or both $i$, then $\ket{1,1}_{TE}$ (if $|\eta_0| = 1$) or $\ket{0,1}_{TE}$ (if $|\eta_1|=1$) never appear in the state following $E$'s application of $\mathcal{F}$.  Note also that, if $|\eta_0|=1$ then, it must hold that $\ket{e_3} = e^{i\theta}\ket{e_1}$ for some $\theta$ (and a similar statement may be made if $|\eta_1| = 1$); this phase change will be done by $\mathcal{F}$ in the forward direction and there is no need to ``create'' the $\ket{e_1}$ state later in the reverse, only the $\ket{e_3}$ state (if $|\eta_1| = 1$, then $\ket{e_2} = e^{i\theta'}\ket{e_0}$ and so there is no need to ``create'' the $\ket{e_2}$ state later, only change the phase which is done by $\mathcal{F}$).  Thus, $V$'s action on these states ($\ket{1,1}$ or $\ket{0,1}$) may be arbitrary and we need not define the corresponding $\ket{g_i}$.  It is clear, then, that $V$ may be made a unitary operator, and the rest of the proof follows as above.

We conclude, therefore, that regardless of $A$ or $B$'s choices, the state of the quantum system for all three parties is the same whether $E$ used $\mathcal{C}$ or $\mathcal{R}$ (meaning the resulting density operator describing the joint systems are equal).  Thus the view from $A$, $B$, and $E$'s point of view are identical in both cases; furthermore, the key-rate computation (Equation \ref{eq:keyrate}) will also be identical in both cases.
\end{proof}

Thus, to prove security of \emph{any} SQKD protocol, it is sufficient to consider only a restricted collective attack - security against general collective attacks will follow from that.  In the next section, we will show how this result may be used for a particular protocol to convert it into an equivalent entanglement based version from which a quantum uncertainty bound may be used to compute the key-rate.

Note that the choice of basis $\mathcal{B}$ is irrelevant to the attacker.  Thus, when analyzing the security of a SQKD protocol using this result, one is free to choose a basis that simplifies the analysis and computations.  In the remainder of this paper, our proof assumes $\mathcal{B} = Z = \{\ket{0}, \ket{1}\}$.  \emph{Also note that our proof would hold even for protocols where $B$ performs a CNOT gate (acting on $\mathcal{H}_T$ and his private register) instead of a projective $Z$ basis measurement} when he chooses $\MR$; mathematically, the two operations will be identical in this case, and the proof above follows through identically.

Before concluding this section, we point out a simplification of our definition if we assume $E$'s attack is symmetric.  This is an assumption often made in quantum security proofs and can even be enforced by the parties.  In particular, if the $Z$ basis error induced by $E$'s attack in the forward channel can be parameterized by a single variable $Q$ (i.e., the probability of an $\ket{i}$ flipping to a $\ket{1-i}$ in the forward channel is $Q$),  and if we work with respect to the $Z$ basis (i.e., $\ket{v_i} = \ket{i}$), then the restricted attack adopts a far simpler form:
\begin{define}\label{def:rest-atk-sym}
A \emph{symmetric restricted collective attack} is a tuple $\mathcal{R}_{\text{sym}} = (Q, \eta, U)$, where $Q \in [0,1]$, $\eta \in \mathbb{D}$, and $U$ is a unitary operator acting on $\mathcal{H}_T\otimes\mathcal{H}_E$.  This attack follows the same process as in Definition \ref{def:rest-atk}, setting $\mathcal{R} = (\sqrt{1-Q}, \sqrt{1-Q}, \eta, -\conj{\eta}, U)$.
\end{define}
It is not difficult to see that if $(q_0, q_1, \eta_0, \eta_1, U)$ is a symmetric attack (i.e., $q_0 = \sqrt{1-Q} = q_1$), then it must hold that $\eta_1 = -\conj{\eta_0}$.

\subsection{An Entanglement Based Protocol}

Our conversion from the prepare-and-measure protocol $\protpm$ to an equivalent, entanglement based one, follows several reductions.  Our goal in this section is to construct a new protocol whereby $B$ (who is no longer classical) prepares quantum states and sends them to $A$ (who is still quantum).  However, by analyzing the security of this new protocol, we will show security of the SQKD protocol $\protpm$.

First, note that $B$'s $\MR$ operation may be equivalently modeled as a CNOT gate acting on the qubit and an empty register private to $B$ \cite{SQKD-second}.  Of course, his $\R$ operation may be modeled as the identity operation.  Thus, when analyzing $\protpm$ we may instead analyze the case where $B$ applies a unitary operation acting on the qubit and a register private to him.  Second, we may assume that $A$, instead of preparing and sending a random state of the form $\ket{0}$, $\ket{1}$, $\ket{+}$, or $\ket{-}$, will instead prepare a Bell state of the form $\frac{1}{\sqrt{2}}(\ket{00} + \ket{11})$, sending one particle to $B$ while keeping the other particle in her private lab.  Standard arguments apply to show that security of this new protocol (which we denote by $\protp$) implies security of the prepare and measure one $\protpm$.  Furthermore, Theorem \ref{thm:rest-atk} still applies (see the comment after the proof).  It is clear that, if we prove security of $\protp$ then security of $\protpm$ will follow.

Now, consider the following protocol which we denote by $\protent$, whose quantum communication stage consists of:
\begin{enumerate}
  \item $B$ prepares the state $\sqrt{p_0}\ket{000}_{A_1A_2B} + \sqrt{p_1}\ket{110}_{A_1A_2B}$ if he wishes to ``$\R$'' otherwise he prepares the state $\sqrt{p_0}\ket{000}_{A_1A_2B} + \sqrt{p_1}\ket{111}_{A_1A_2B}$ if he wishes to ``$\MR$'' (he will chose the operation randomly each iteration).  The two qubits $A_1$ and $A_2$ are sent to Alice.  Note that the terminology ``$\MR$'' and ``$\R$'' no longer has any real meaning in this protocol.
  \item $A$ receives both particles $A_1$ and $A_2$ and will measure each in the $Z$ basis or the $X$ basis, choosing randomly.
  \item If $B$ chooses to $\MR$ and if $A$ uses the $Z$ basis, they may use their results as their raw-key bit (we assume $A_1$ is used as $A$'s raw-key bit, though our analysis below would be symmetric if $A_2$ were used instead).  Other iterations, along with a random subset of these ``raw-key'' iterations,  may be used for error estimation in the obvious way.
\end{enumerate}

We give $E$ the ability to control the setting of $p_0$ and $p_1$ which can only increase her power (and, as a consequence, also gives us partial device independence - indeed, one can consider the scenario that $E$ manufactured the device $B$ is using and programmed in a particular $p_0$ and $p_1$ value).  A collective attack against this protocol, thus, is a setting for $p_0$ (the value $p_1 = 1-p_0$ of course) and a unitary operator $U$ acting on two qubits $A_1$ and $A_2$ and $E$'s private quantum memory.

While $\protent$ is not a ``true'' entanglement-based version (as $B$ is making a choice between two pure states), it would not be difficult to make it one simply by increasing the dimension of $B$'s space with an extra qubit (which, after measuring, would determine his choice of $\MR$ or $\R$).  However, as it turns out, the protocol as described will be sufficient to complete our security analysis of the prepare-and-measure protocol.

We claim that, if $\protent$ is secure, then so is $\protp$ (in which case, so is $\protpm$).  In particular, we will show that, for any attack against $\protp$, there exists an attack against $\protent$ which exactly replicates the resulting quantum system.  In particular, given an attack against $\protp$, we will construct an attack against $\protent$ which first ``rewinds'' the forward channel attack simulating the system had $A$ initially sent a qubit as opposed to $B$.  Informally, as an example, if $B$ sends a $\ket{0}$ in the $A_1$ register (which he never does - it is always a pair of qubits but this is simply for illustration), then we construct an attack which causes $A$ to receive a $\ket{0}$ or $\ket{1}$ in her $A_1$ register with the same probabilities as if she had sent a $\ket{0}$ or a $\ket{1}$ and $B$ happened to measure a $\ket{0}$ if they were running $\protp$.  Furthermore, $E$'s memory will be in the same state in both cases.  Thus, we will ``rewind'' the forward channel for $\protent$ to simulate the scenario where $A$ sends a qubit first instead of $B$ sending a qubit first.  The only thing $E$ cannot ``rewind'' is the probability of $B$ observing certain outcomes in $\protp$ - thus the need for her to set the value of $p_0$ during device construction.

\begin{theorem}
Let $(U_F, U_R)$ be a collective attack used against $\protp$ and let $\rho_{ABE}$ be the density operator describing a single iteration of this protocol $\protp$ when $E$ uses this attack.  Then, there exists an attack $\mathcal{E} = (p_0,U)$ against $\protent$ such that, if $\sigma_{ABE}$ is the resulting density operator when running $\protent$ using attack $\mathcal{E}$, it holds that $\rho_{ABE} = \sigma_{ABE}$ assuming the probability that $B$ chooses $\MR$ in $\protp$ is the same as in $\protent$ (and thus the probability of choosing $\R$ is also the same in both protocols).
\end{theorem}
\begin{proof}
Let $(U_F, U_R)$ be a collective attack against $\protp$.  Since Theorem \ref{thm:rest-atk} applies, there exists an equivalent restricted attack consisting of the forward operator $\mathcal{F}$ as described in Equation \ref{eq:rest-atk-forward}.  We construct the desired attack $\mathcal{E}$ against $\protent$.

Consider the following operator $\rew$ to be used against $\protent$ in order to ``rewind'' the forward channel attack.  The action of this operator is:
\begin{align}
\rew \ket{00}_{A_1A_2} &= \frac{q_0\ket{000}_{A_1A_2E} + \sqrt{1-q_1^2}\ket{10f}_{A_1A_2E}}{\sqrt{1-q_1^2+q_0^2}}\label{eq:rew-action}\\
\rew \ket{11}_{A_1A_2} &= \frac{\sqrt{1-q_0^2}\ket{01e}_{A_1A_2E} + q_1\ket{110}_{A_1A_2E}}{\sqrt{1-q_0^2+q_1^2}},\notag
\end{align}
where $\ket{e}$ and $\ket{f}$ are the states resulting from the application of $\mathcal{F}$ (see Equations \ref{eq:rest-atk-e} and \ref{eq:rest-atk-f}).  We claim that $\rew$ is an isometry (and thus can be extended to a unitary operation).  This is not difficult to see: let $\ket{v_0} = \rew\ket{00}$ and $\ket{v_1} = \rew\ket{11}$.  It is clear that $\braket{v_0|v_1} = \braket{v_1|v_0} = 0$.  Furthermore:
\[
\braket{v_0|v_0} = \frac{1}{1-q_1^2+q_0^2}(q_0^2+1-q_1^2) = 1,
\]
and similarly $\braket{v_1|v_1} = 1$.  Thus, $\rew$ is an isometry from $\mathcal{H}_{A_1A_2} \rightarrow\mathcal{H}_{A_1A_2E}$.  We abuse notation from here and assume $\rew$ is a unitary operator (its action on other states may be arbitrary).

Now, consider the following attack against $\protent$: $\mathcal{E} = (p_0, (I_{A_1}\otimes U_R)\cdot\rew)$, where $p_0 = \frac{1}{2}(1 - q_1^2+q_0^2)$ (and, so, $p_1 = 1-p_0 = \frac{1}{2}(1-q_0^2+q_1^2)$).  We claim this is the desired attack.  Indeed, if $B$ chooses to $\R$ in $\protp$, then the state of the system after the qubit leaves $B$'s lab, but before $E$ applies $U_R$ is:
\begin{align*}
\ket{\psi^R_{A_1A_2BE}} &= \frac{1}{\sqrt{2}}\left(q_0\ket{0000} + \sqrt{1-q_0^2}\ket{010e}\right)\\
 &+ \frac{1}{\sqrt{2}}\left(\sqrt{1-q_1^2}\ket{100f} + q_1\ket{1100}\right),
\end{align*}
where the order of the systems on the right-hand-side of the above expression are: $A_1A_2BE$ (note $B$ is un-entangled from the above system).  At this point, $E$ has control only of the $A_2E$ subspace.

On the other hand, if $B$ chooses $\R$ in $\protent$, then the state of the system after $E$ applies $\rew$ but before finishing the attack with $U_R$, is:
\begin{align}
&\ket{\phi^R_{A_1A_2BE}} = \label{eq:phiR}\\
& \frac{\sqrt{p_0}}{\sqrt{1-q_1^2+q_0^2}}\left(q_0\ket{0000} + \sqrt{1-q_1^2}\ket{100f}\right)\notag\\
+&\frac{\sqrt{p_1}}{\sqrt{1-q_0^2+q_1^2}}\left(\sqrt{1-q_0^2}\ket{010e} + q_1\ket{1100}\right)\notag\\\notag\\
=& \frac{1}{\sqrt{2}}\left(q_0\ket{0000} + \sqrt{1-q_1^2}\ket{100f}\right)\notag\\
+&\frac{1}{\sqrt{2}}\left(\sqrt{1-q_0^2}\ket{010e} + q_1\ket{1100}\right),\notag
\end{align}
from which it is clear that $\ket{\psi^R_{A_1A_2BE}} \equiv \ket{\phi^R_{A_1A_2BE}}$.  With $\protp$, $E$ will then apply $U_R$ (acting on $A_2E$) to $\ket{\psi^R}$ and then forward the $A_2$ system to Alice; with $\protent$, $E$ will apply the same $U_R$ (also acting on the subspace $A_2E$) to $\ket{\phi^R}$ and forwards both $A_1$ and $A_2$ to Alice.  Regardless, we find the quantum system held by all three parties to be equal.

The case when $B$ chooses to $\MR$ is similar.  Indeed, in this case, consider the state before $E$ applies $U_R$ for $\protp$:
\begin{align*}
\ket{\psi^M_{A_1A_2BE}} &= \frac{1}{\sqrt{2}}\left(q_0\ket{0000} + \sqrt{1-q_0^2}\ket{011e}\right)\\
 &+ \frac{1}{\sqrt{2}}\left(\sqrt{1-q_1^2}\ket{100f} + q_1\ket{1110}\right).
\end{align*}
And, the case for $\protent$ after $E$ applies $\rew$ but before $U_R$ is:
\begin{align}
&\ket{\phi^M_{A_1A_2BE}} = \label{eq:phiM}\\
& \frac{\sqrt{p_0}}{\sqrt{1-q_1^2+q_0^2}}\left(q_0\ket{0000} + \sqrt{1-q_1^2}\ket{100f}\right)\notag\\
+&\frac{\sqrt{p_1}}{\sqrt{1-q_0^2+q_1^2}}\left(\sqrt{1-q_0^2}\ket{011e} + q_1\ket{1110}\right)\notag\\\notag\\
=& \frac{1}{\sqrt{2}}\left(q_0\ket{0000} + \sqrt{1-q_1^2}\ket{100f}\right)\notag\\
+&\frac{1}{\sqrt{2}}\left(\sqrt{1-q_0^2}\ket{011e} + q_1\ket{1110}\right).\notag
\end{align}
Again, we conclude the two systems will be the same after application of $U_R$.  The final density operator will then be a mixture of the two pure states.  Assuming the probability that $B$ chooses to $\MR$ in $\protp$ is equal to the probability he chooses this option in $\protent$ (and thus, the probability of choosing $\R$ is also equal in both protocols) the resulting density operators will be identical.
\end{proof}

Notice that, when attacking $\protent$, $E$ has control of both $A_1$ and $A_2$ (which she does not when attacking $\protp$).  Thus, she has possibly more attack strategies against $\protent$.  However, for every attack against $\protp$, there exists an equivalent attack against $\protent$.  Thus, to prove security of $\protp$ (and thus $\protpm$, our goal), it suffices to analyze $\protent$ as $E$ has potentially more attack capabilities against the latter.  Indeed, we have the following ``chain:''
\[
\protent \Longrightarrow \protent_{res} \Longrightarrow \protp \Longrightarrow \protpm,
\]
where $\protent_{res}$ is the protocol $\protent$ but with $E$ restricted to attacks of the form $\mathcal{E} = (p_0, (I_{A_1}\otimes U)\cdot\rew)$.

\subsection{Final Key-Rate Bound}

Consider the protocol $\protent$ introduced in the previous subsection.  There are two ``modes'' to it: either $B$ chooses to $\MR$ (with probability $P_M$) or he chooses to $\R$ (with probability $P_R)$ - not that these terms have the same meaning as their wording implies.  A single iteration of the protocol, then, may be written as the density operator:
\[
\rho_{A_1A_2BE} = P_M\sigma_{A_1A_2BE} + P_R\tau_{A_1A_2BE},
\]
where:
\begin{align}
\sigma_{A_1A_2BE} = \kb{\phi^M_{A_1A_2BE}}, && \tau_{A_1A_2BE} = \kb{\phi^R_{A_1A_2BE}}\label{eq:state-tau}
\end{align}
and $\ket{\phi^M}$ and $\ket{\phi^R}$ are the (pure) states in the event $B$ chooses $\MR$ or $\R$ respectively.  Now, only those iterations where $B$ chooses $\MR$ are used for key distillation and, so, to compute the key-rate of this protocol, we need to compute $S(A_1^Z|E)_\sigma$ where we use $A_1^Z$ to denote the result of $A$ measuring her $A_1$ register in the $Z$ basis (recall $A$ uses only $A_1$ for her key distillation).  However, we will first analyze $S(A_1|E)_\tau$ and use this to bound the entropy in $\sigma$.


It was shown in \cite{QKD-uncertainty}, that for any density operator acting on a tripartite Hilbert space $\mathcal{H}_A\otimes\mathcal{H}_B\otimes\mathcal{H}_E$, that if $A$ and $B$ make measurements in the $Z$ or $X$ basis, then:
\begin{equation}\label{eq:uncertain}
S(A^Z|E) + S(A^X|B) \ge 1,
\end{equation}
where we use $A^Z$ (respectively $A^X$) to denote the register storing the result of a $Z$ (respectively $X$) basis measurement on the $A$ system.  Using this, we may easily prove the following:
\begin{lemma}
Let $\tau_{A_1A_2BE}$ be the state of the system if $B$ chooses $\R$ in protocol $\protent$ and let $Q_X$ be the error rate in the $X$ basis between registers $A_1$ and $A_2$.  Then:
\begin{equation}
S(A_1^Z|E)_\tau \ge 1 - h(Q_X).
\end{equation}
\end{lemma}
\begin{proof}
Note that $B$ is completely independent of the $A_1A_2E$ system; i.e., $\tau_{A_1A_2BE} \equiv \tau_{A_1A_2E} \otimes \kb{0}_B$.  Thus, we may simply consider the state resulting from tracing out $B$ which acts on the tripartite system $\mathcal{H}_{A_1}\otimes\mathcal{H}_{A_2}\otimes\mathcal{H}_E$.  Using the uncertainty relation described above (see Equation \ref{eq:uncertain}), replacing $B$ with $A_2$ (thus, in a way, we are imagining $A$ as two people - one who holds the $A_1$ register and the other who holds the $A_2$ register - of course in real life they are one individual), we have:
\begin{align*}
S(A_1^Z|E) &\ge 1 - S(A_1^X|A_2)\\
&\ge 1-H(A_1^X|A_2^X)\\
&= 1-h(Q_X),
\end{align*}
where the last inequality follows from the fact that measurements can only increase entropy.
\end{proof}

We will use the conditional entropy of $\tau_{A_1E}$ to compute a bound on the entropy in $\sigma_{A_1E}$, thus giving us our desired key-rate computation.  For the following result, we will assume a symmetric attack in that the observed $Z$ basis noise is equal in both forward and reverse channels; this is only to make the algebra more amicable - analyzing an asymmetric channel would follow the same process, just with slightly more, yet trivial, algebra.

Before we continue, however, we require one lemma which, though straight-forward to show, we include for completeness:
\begin{lemma}\label{lemma:trd-bound}
Let $\sigma = \ket{v_0}\bra{v_1} + \ket{v_1}\bra{v_0}$.  Then:
\[
\trd{\sigma} \le 2\sqrt{\braket{v_0|v_0}\braket{v_1|v_1}}.
\]
\end{lemma}
\begin{proof}
Recall that the trace norm is invariant to unitary changes in basis.  We decompose $\ket{v_0}$ and $\ket{v_1}$ as:
\begin{align*}
\ket{v_0} &= x\ket{E}, && \ket{v_1} = he^{i\theta}\ket{E} + d\ket{I},
\end{align*}
where $\braket{E|E} = \braket{I|I} = 1$ and $\braket{E|I} = 0$.  Furthermore, we may assume $x,h,d\in\mathbb{R}$ (any alternative phase of $x$ or $d$ may be absorbed into the corresponding basis vector).  In this $\{\ket{E}, \ket{I}\}$ basis, we have:
\[
\sigma = \left(\begin{array}{cc}
xhe^{i\theta}+xhe^{-i\theta} & xd\\
xd & 0\end{array}\right) = 
\left(\begin{array}{cc}
2xRe\left(he^{i\theta}\right) & xd\\
xd & 0\end{array}\right),
\]
the eigenvalues of which are:
\begin{align*}
\lambda_\pm &= xRe\left(he^{i\theta}\right) \pm \sqrt{x^2Re^2\left(he^{i\theta}\right)+x^2d^2}.
\end{align*}
Writing $\braket{v_0|v_1} = a + bi$, with $a,b\in \mathbb{R}$, we have the following identities (which follows from the decomposition of $\ket{v_0}$ and $\ket{v_1}$):
\begin{align}
&x^2 = \braket{v_0|v_0}\\
&he^{i\theta} = \frac{\braket{v_0|v_1}}{x} = \frac{a}{x} + \frac{b}{x}i\Rightarrow h^2 = \frac{a^2}{x^2} + \frac{b^2}{x^2}\\
&h^2 + d^2 = \braket{v_1|v_1}\Rightarrow d^2 = \braket{v_1|v_1} - \frac{a^2}{x^2} - \frac{b^2}{x^2}.
\end{align}
Substituting these into $\lambda_\pm$ yields:
\begin{align*}
\lambda_\pm &= xRe\left(he^{i\theta}\right) \pm \sqrt{x^2Re^2\left(he^{i\theta}\right)+x^2d^2}\\
&= a \pm \sqrt{a^2 + x^2\braket{v_1|v_1} - a^2 - b^2}\\
&= a \pm \sqrt{\braket{v_0|v_0}\braket{v_1|v_1} - b^2}.
\end{align*}
The Cauchy-Schwarz inequality forces $b^2 \le \braket{v_0|v_0}\braket{v_1|v_1}$, so the square-root is real.  In fact, by the Cauchy-Schwarz inequality:
\begin{equation}\label{eq:lemma-bound1}
a^2+b^2 \le \braket{v_0|v_0}\braket{v_1|v_1} \Rightarrow |a| \le \sqrt{\braket{v_0|v_0}\braket{v_1|v_1}-b^2}.
\end{equation}
Now: $\trd{\sigma} = |\lambda_+|+|\lambda_-|$.  If $a \ge 0$, then, using Equation \ref{eq:lemma-bound1} and letting $\delta = \sqrt{\braket{v_0|v_0}\braket{v_1|v_1}-b^2}$ (thus $0\le a \le \delta$):
\[
\trd{\sigma} = |a+\delta| + |a-\delta| = (a+\delta) + (\delta-a) = 2\delta.
\]
Alternatively, if $a < 0$:
\[
\trd{\sigma} = |a+\delta| + |a-\delta| = (a+\delta) + (\delta-a) = 2\delta.
\]
Of course $\delta \le \sqrt{\braket{v_0|v_0}\braket{v_1|v_1}}$ completing the proof.
\end{proof}

Finally, we prove the following theorem which bounds the von Neumann entropy in $\sigma$ allowing us to compute the key-rate of this SQKD protocol.

\begin{theorem}\label{thm:main-ent}
Given $\sigma_{A_1A_2BE}$ and $\tau_{A_1A_2BE}$ as defined above in Equation \ref{eq:state-tau}, let $Q$ be the error rate in the $Z$ basis observed on a single channel (we assume both channels have the same $Z$ basis error rate). Also, let $\delta$ be defined as:
\[
\delta = 2Q(1-Q) + \left(\frac{1}{2} + 2Q[1-Q]\right)\cdot h\left(\frac{4Q(1-Q)}{1+4Q(1-Q)}\right).
\]
Then, assuming $E$'s attack is symmetric and of the form $\mathcal{E} = (1/2, (I_{A_1}\otimes V)\rew)$, where $V$ acts on $\mathcal{H}_{A_2E}$ and $\rew$ is a ``rewind'' operator as discussed earlier, it holds that $S(A_1^Z|E)_\sigma \ge f(Q)$, where:
\begin{equation}
f(Q) = \left\{\begin{array}{ll}
S(A_1^Z|E)_\tau - \delta & \text{ if } S(A_1^Z|E)_\tau \ge 2\delta\\
\frac{1}{2}S(A_1^Z|E)_\tau & \text{ otherwise}
\end{array}\right.
\end{equation}
\end{theorem}
\begin{proof}
Let $U$ be an arbitrary attack operator used against $\protent$ (this is an isometry from $\mathcal{H}_{A_1A_2}$ to $\mathcal{H}_{A_1A_2E}$).  Also, let $\ket{\phi_\pm} = \frac{1}{\sqrt{2}}(\ket{00} \pm \ket{11})$ be two Bell states.  Without loss of generality, we may write $U$'s action as:
\begin{align}
U\ket{\phi_+} &= \sum_{a,b\in\{0,1\}}\ket{a,b}_{A_1A_2}\otimes\ket{e_{ab}}_E\label{eq:general-U-action}\\
U\ket{\phi_-} &= \sum_{a,b\in\{0,1\}}\ket{a,b}_{A_1A_2}\otimes\ket{f_{ab}}_E,\notag
\end{align}
where the states $\ket{e_{ab}}$ and $\ket{f_{ab}}$ are arbitrary, not necessarily normalized nor orthogonal, states in $\mathcal{H}_E$.  The density operator $\tau_{A_1^ZA_2^ZE}$ (which is the state of the system when $B$ chooses $\R$ in protocol $\protent$ and $A$ measures both her qubits in the $Z$ basis - an operation denoted by $\mathcal{M}$ below) when faced with this attack is found to be:
\begin{align*}
&\tau_{A_1^ZA_2^ZE} = \mathcal{M}\left[U\kb{\phi_+}U^*\right] = \sum_{a,b}\kb{a,b}\otimes\kb{e_{ab}}\\
\Rightarrow&\tau_{A_1^ZE} = \sum_{a\in\{0,1\}}\kb{a}\otimes(\kb{e_{a0}}+\kb{e_{a1}}).
\end{align*}
Likewise, we may compute $\sigma$, the density operator for those iterations where $B$ chooses $\MR$ (note that, below, we define $P(z) = zz^*$):
\begin{align*}
&\sigma_{A_1^ZA_2^ZE} = \mathcal{M}\left[U\left(\frac{1}{2}\kb{00} + \frac{1}{2}\kb{11})\right)U^*\right] \\
&=\mathcal{M}\left[\frac{1}{4}P(U\ket{\phi_+}+U\ket{\phi_-}) + \frac{1}{4}P(U\ket{\phi_+} - U\ket{\phi_-})\right]\\
&=\frac{1}{4}\sum_{a,b}\kb{a,b}\otimes\left(P[\ket{e_{a,b}} + \ket{f_{a,b}}] + P[\ket{e_{a,b}}-\ket{f_{a,b}}]\right)\\
&=\frac{1}{2}\sum_{a,b}\kb{a,b}\otimes(\kb{e_{a,b}} + \kb{f_{a,b}}).
\end{align*}
From the above, it is not difficult to see:
\begin{equation}
\sigma_{A_1^ZE} = \frac{1}{2}\tau_{A_1^ZE} + \frac{1}{2}\mu_{A_1^ZE},
\end{equation}
where:
\begin{equation}
\mu_{A_1^ZE} = \sum_{a}\kb{a}\otimes(\kb{f_{a0}}+\kb{f_{a1}}).
\end{equation}
Thus, by concavity of conditional entropy, we have:
\begin{equation}\label{eq:ent-sigma-bound}
S(A_1^Z|E)_\sigma \ge \frac{1}{2}S(A_1^Z|E)_\tau + \frac{1}{2}S(A_1^Z|E)_\mu
\end{equation}
Recall that we are actually interested in proving security against $\protp$ and so only need to concern ourselves with attacks of the form: $U = (I_{A_1}\otimes V)\rew$.  From Equation \ref{eq:rew-action}, and noting that, since we are assuming a symmetric attack in that $p_0 = p_1 = 1/2$ and so $q_0^2=q_1^2$, the action of $\rew$ on $\ket{\phi_\pm}$ is:
\[
\rew\ket{\phi_\pm} = \frac{1}{\sqrt{2}}(q_0\ket{000} + \bar{q_1}\ket{10f} \pm \bar{q_0}\ket{01e} \pm q_1\ket{110}),
\]
where we use $\bar{q_0}$ to mean $\sqrt{1-q_0^2}$ (and similarly for $\bar{q_1}$) and where $\ket{e}$ and $\ket{f}$ are defined in Equations \ref{eq:rest-atk-e} and \ref{eq:rest-atk-f}.

Now, we may, without loss of generality, describe $V$'s action as follows (recall $V$ is a unitary operator acting on $\mathcal{H}_{A_2E}$):
\begin{align*}
V\ket{00} = \ket{0,e_0}+\ket{1,e_1}, && V\ket{10} = \ket{0,e_2} + \ket{1,e_3}\\
V\ket{0f} = \ket{0,f_0}+\ket{1,f_1}, && V\ket{1e} = \ket{0,f_2} + \ket{1,f_3}
\end{align*}
Thus, after applying attack $U = (I_{A_1}\otimes V)\cdot\rew$ to $\ket{\phi_\pm}$, we find:
\begin{align*}
\sqrt{2}U\ket{\phi_{\pm}} &= q_0\ket{00e_0} + q_0\ket{01e_1} + \bar{q_1}\ket{10f_0} + \bar{q_1}\ket{11f_1}\\
&\pm\bar{q_0}\ket{00f_2} \pm \bar{q_0}\ket{01f_3} \pm q_1\ket{10e_2} \pm q_1\ket{11e_3}\\
\end{align*}
Translating to notation used in Equation \ref{eq:general-U-action}, we find:
\begin{align*}
\ket{e_{00}} &= \frac{1}{\sqrt{2}}\left(q_0\ket{e_0} + \bar{q_0}\ket{f_2}\right)\\
\ket{e_{01}} &= \frac{1}{\sqrt{2}}\left(q_0\ket{e_1} + \bar{q_0}\ket{f_3}\right)\\
\ket{e_{10}} &= \frac{1}{\sqrt{2}}\left(\bar{q_1}\ket{f_0} + q_1\ket{e_2}\right)\\
\ket{e_{11}} &= \frac{1}{\sqrt{2}}\left(\bar{q_1}\ket{f_1} + q_1\ket{e_3}\right),
\end{align*}
while the states $\ket{f_{ab}}$ are found simply by taking $\ket{e_{ab}}$ and changing the ``$+$'' to a ``$-$'' in between the two kets.

Our goal is to determine a bound on $S(A_1^Z|E)_\mu$ using $S(A_1^Z|E)_\tau$.  To do so, we will use a continuity bound on conditional entropy determined by Winter in \cite{QC-winter2016tight} (a tighter version of the Alicki-Fannes inequality \cite{QC-alicki2004continuity}).  Using this bound, we have:
\begin{equation}\label{eq:ent-bound}
|S(A_1^Z|E)_\tau - S(A_1^Z|E)_\mu| \le \epsilon + (1+\epsilon)h\left(\frac{\epsilon}{1+\epsilon}\right),
\end{equation}
where:
\[
\frac{1}{2}\trd{\tau_{A_1^ZE} - \mu_{A_1^ZE}} \le \epsilon \le 1.
\]
Thus, if we bound the trace distance between $\tau$ and $\mu$ we can determine an upper-bound on $S(A_1^Z|E)_\mu$ thus giving us our result.

By the triangle inequality, we have:
\begin{align*}
\trd{\tau_{A_1^ZE} - \mu_{A_1^ZE}} &\le \sum_{a,b}\trd{\kb{e_{ab}} - \kb{f_{ab}}}\\
&=q_0\bar{q_0}\trd{\ket{e_0}\bra{f_2} + \ket{f_2}\bra{e_0}}\\
&+q_0\bar{q_0}\trd{\ket{e_1}\bra{f_3} + \ket{f_3}\bra{e_1}}\\
&+\bar{q_1}q_1\trd{\ket{f_0}\bra{e_2} + \ket{e_2}\bra{f_0}}\\
&+\bar{q_1}q_1\trd{\ket{f_1}\bra{e_3} + \ket{e_3}\bra{f_1}}\\\\
&\le 2q_0\bar{q_0}\sqrt{\braket{e_0|e_0}\braket{f_2|f_2}}\\
&+ 2q_0\bar{q_0}\sqrt{\braket{e_1|e_1}\braket{f_3|f_3}}\\
&+ 2\bar{q_1}q_1\sqrt{\braket{e_2|e_2}\braket{f_0|f_0}}\\
&+ 2\bar{q_1}q_1\sqrt{\braket{e_3|e_3}\braket{f_1|f_1}},
\end{align*}
where the last inequality follows from Lemma \ref{lemma:trd-bound}.

Assuming the attack is symmetric, with $Q$ being the $Z$ basis noise in each channel (i.e., the probability of a $\ket{i}$ flipping to a $\ket{1-i}$ in either channel is $Q$), then it is easy to see that $q_0 = q_1 = \sqrt{1-Q}$ (thus $\bar{q_1} = \bar{q_0} = \sqrt{Q}$) and also:
\begin{align*}
&\braket{e_0|e_0} = \braket{e_3|e_3} = \braket{f_0|f_0} = \braket{f_3|f_3} = (1-Q)\\
&\braket{e_1|e_1} = \braket{e_2|e_2} = \braket{f_1|f_1} = \braket{f_2|f_2} = Q
\end{align*}
(Note that these values are observable by $A$ and $B$ in both $\protent$ and $\protpm$.)

Thus:
\begin{align}
\frac{1}{2}\trd{\tau_{A_1^ZE} - \mu_{A_1^ZE}} &\le \frac{1}{2}\left(8\sqrt{Q(1-Q)}\sqrt{Q(1-Q)}\right)\\
&=4Q(1-Q).
\end{align}
From Equation \ref{eq:ent-bound}, and setting $\epsilon = 4Q(1-Q)$, we have:
\[
S(A_1^Z|E)_\mu \ge S(A_1^Z|E)_\tau - \epsilon - (1+\epsilon)h\left(\frac{\epsilon}{1+\epsilon}\right).
\]
Combining this with Equation \ref{eq:ent-sigma-bound}, and also noting that, since $\mu_{A_1^ZE}$ is a classical-quantum state, and so $S(A_1^Z|E)_\mu \ge 0$, completes the proof.
\end{proof}

Note that the only place we used a symmetric assumption was in the computation of $q_i$, $\braket{e_i|e_i}$, and $\braket{f_i|f_i}$ and in the assumption that $B$ sends Bell states $\ket{\phi_{\pm}}$; if the attack was asymmetric, the only thing that would change would be these quantities - however our technique used in the proof of Theorem \ref{thm:main-ent} would still apply, it would just require  more algebra.  Note, however, that since these are observable quantities, $A$ and $B$ can even enforce that $E$ use a symmetric attack.

The key-rate of the protocol $\protent_{res}$, therefore, is:
\begin{align}
r &= S(A_1^Z|E)_\sigma - H(A_1^Z|B^Z)_\sigma\notag\\
&\ge g(Q,Q_X) - h(Q),
\end{align}
where:
\[
g(Q,Q_X) = \left\{\begin{array}{ll}
1-h(Q_X) - \delta & \text{ if } 1-h(Q_X) \ge 2\delta\\
\frac{1}{2}(1-h(Q_X)) & \text{ otherwise}
\end{array}\right.
\]
and $\delta$ is defined in Theorem \ref{thm:main-ent}.  Note that, above, we used the trivial fact that $H(A_1^Z|B^Z)_\sigma = h(Q)$ in this scenario (an asymmetric attack will be different, but also easily computed).  Of course, by our analysis conducted above, the key-rate of $\protpm$ can only be higher than this.

We now evaluate our new key-rate bound, in particular its noise tolerance, comparing to current known results.  In \cite{SQKD-Krawec-SecurityProof}, a lower bound on the key-rate of $\protpm$ was derived without mismatched measurements.  In \cite{QKD-Tom-Krawec-Arbitrary}, we derived a new bound, but using numerous mismatched measurement statistics (e.g., the probability of a $\ket{+}$ being measured as a $\ket{0}$).  Such statistics give greater information on the attack being used (thus increasing the bound on $S(A|E)$) but at the cost of wasting quantum communication on error estimation.

As shown in Table \ref{table:results}, the noise tolerance of our new bound here is higher than previous best work utilizing only noise statistics (and not mismatched measurements) \cite{SQKD-Krawec-SecurityProof}.  However, it is lower than the bound derived using mismatched measurements as we did in \cite{QKD-Tom-Krawec-Arbitrary}.  \emph{Note that we did not require any mismatched measurements for our derivation in this paper,} thus it is expected that our key-rate bound may not be as high as that discovered in \cite{QKD-Tom-Krawec-Arbitrary}.  In that work, we used over 12 different measurement statistics to bound $S(A|E)$.  In this security proof, we are only using \emph{three} measurement statistics: the error rate in the $Z$ basis (for both channels) and the $X$ error rate observed by $A$.  Thus, it is not surprising our noise tolerance here is lower. An interesting open question is whether this difference is primarily an artifact of our proof (in particular, Theorem \ref{thm:main-ent} may not be tight) or whether this shows the necessity of using mismatched statistics for this protocol (such statistics are known to be highly beneficial to limited-resource protocols \cite{QKD-Tom-threestate1} and even required for some SQKD protocols \cite{SQKD-krawec2017limited}).  We intend to investigate this further in the near future.

\begin{table}
\begin{tabular}{l|c|c|c}
&Old Proof \cite{SQKD-Krawec-SecurityProof} & New Proof & w/MM \cite{QKD-Tom-Krawec-Arbitrary}\\
\hline
$Q_X = Q$ & $5.34\%$ & $\mathbf{6.14\%}$ & $11\%$\\
$Q_X = 2Q(1-Q)$ & $4.57\%$ & $\mathbf{4.82\%}$ & $7.9\%$\\
$Q_X = \frac{1}{2}Q$ & $5.92\%$ & $\mathbf{7.5\%}$ & $15.12\%$
\end{tabular}
\caption{Showing the maximal noise tolerance (i.e., the maximum $Q$ for which the key-rate of $\protpm$ remains positive) of $\protpm$ derived here (middle column labeled ``New Proof'').  Also comparing with current known results for this protocol.  ``Old Proof'' is from \cite{SQKD-Krawec-SecurityProof} and bounded the key-rate using only three noise statistics (as was done here).  ``w/MM'' was from \cite{QKD-Tom-Krawec-Arbitrary} and utilized mismatched measurements (in total over 12 different measurement statistics as opposed to only three used here).  Our new result is higher than that in \cite{SQKD-Krawec-SecurityProof}, but not as high as when utilizing mismatched measurement statistics. See text for further explanation.}\label{table:results}
\end{table}

\subsection{Comment on General Attacks}

Normally, one may extend the computations done for collective attacks to produce security against general attacks, in the asymptotic scenario, by invoking a de Finetti type argument \cite{QKD-renner-keyrate,QKD-general-attack,QKD-general-attack2}.  It seems like this should also hold for our security proof here, however, due to our reliance on the restricted collective attack, this would require a more rigorous proof\footnote{Thanks to Rotem Liss for our conversations on this.}.  We suspect the results in this paper would carry through to imply security against general attacks, however we leave a formal, rigorous, proof of this as future work.

\section{Closing Remarks}

In this paper, we presented an entirely new proof of security for a semi-quantum protocol.  By first reducing the problem to an equivalent entanglement-based protocol and then using a quantum entropic uncertainty relation, along with a continuity bound on conditional von Neumann entropy, we were able to derive a cleaner key-rate expression for this protocol than previous work.  Furthermore, our new key-rate has a higher noise tolerance than previous work \emph{without mismatched measurements}.  Of course, our new key-rate has a lower tolerance than when using mismatched measurements.  An open question worth investigating is whether this is an artifact of our proof technique or if mismatched measurements are absolutely required for this protocol to achieve optimal noise tolerance.

This new technique we developed in this paper, along with the various security results produced along the way, such as the restricted attack, may hold great application in studying other protocols relying on a two-way channel.  It may also hold the key to studying SQKD protocols relying on higher-dimensional quantum channels such as the one proposed in \cite{QKD-vlachou2017quantum} based on quantum walks.  We suspect this technique can be suitably adapted to handle other two-way protocols and other semi-quantum protocols in higher dimensions.


\end{document}